\documentclass[letterpaper, 10 pt, conference]{ieeeconf}  

\IEEEoverridecommandlockouts                              
\usepackage{color}
\usepackage{amsfonts}
\usepackage{amsmath}
\usepackage{graphicx}
\usepackage{algorithm,algorithmicx,algpseudocode}
\usepackage{hyperref}
\usepackage{mathtools}
\usepackage{subcaption}

\title{\LARGE \bf
Cause Mining and Controller Synthesis with STL
}

\author{Irmak Saglam and Ebru Aydin Gol
\thanks{This work has received funding from the European Union's Horizon 2020 research and innovation programme under the Marie Sk\l{}odowska-Curie grant agreement No 798482 and from the Scientific and Technological Research Council of Turkey (TUBITAK) under the TUBITAK project number 117E242.}
\thanks{Irmak Saglam and Ebru Aydin Gol are with the Department of Computer Engineering,
      Middle East Technical University, Ankara/TURKEY
        {\tt\small \{saglam.irmak, ebrugol\}}{\tt\small @metu.edu.tr}}%
}

\newcommand{\pastF}{\ensuremath{F^-}}
\newcommand{\pastG}{\ensuremath{G^-}}
\newcommand{\since}{\ensuremath{S}}

\newtheorem{proposition}{\bf Proposition}[section]

\newtheorem{example}{\bf Example}[section]

\newtheorem{problem}{\bf Problem}[section] 

\begin{document}

\maketitle
\thispagestyle{empty}
\pagestyle{empty}

\begin{abstract}
Formal control of cyber-physical systems allows for synthesis of control strategies from rich specifications such as temporal logics. However, the classes of systems that the formal approaches can be applied to is limited due to the computational complexity. Furthermore, the synthesis problem becomes even harder when non-determinism or stochasticity is considered. In this work, we propose an alternative approach. First, we mark the unwanted events on the traces of the system and generate a \textit{controllable cause} representing these events as a Signal Temporal Logic (STL) formula. Then, we synthesize a controller based on this formula to avoid the satisfaction of it. Our approach is applicable to any system with finitely many control choices. While we can not guarantee correctness we show on an example that the proposed approach reduces the number of the unwanted events. In particular,  we validate it for the congestion avoidance problem in a traffic network. 

\end{abstract}

\section{INTRODUCTION}


In formal control, the goal is to synthesize control strategies from formal specifications expressed in rich specification languages. The main approach in formal control of dynamical systems is based on construction of a finite abstraction of the system, and then synthesis of a control strategy for the abstract model from the specification~\cite{belta2017}. Finally, the strategy is mapped to the original system. This approach has been applied to linear, switched, and piecewise-affine, and hybrid systems~\cite{tabuada2006linear,belta2017,Kavraki:MPlanning}. In addition, game theoretic and probabilistic versions of the formal control problem has been studied for non-deterministic and stochastic systems~\cite{zamani2014symbolic,belta2017,KlBe-HSCC08-book}. Despite the promising results on synthesis of correct-by-construction controllers, the abstraction based approach suffers from scalability issues due to the complex operations involved in the abstraction process and the size of the resulting abstract model. Furthermore, the developed methods are specific to the underlying system and extensions to more complex systems such as systems with a variety of submodules, e.g. Simulink models~\cite{Simulink}, is not straightforward. 

In this work, we propose an alternative approach that is applicable to any discrete-time dynamical system with finitely many control choices. 
In particular, we consider the following problem: given a discrete time control system with a finite control set, a function over the state and control spaces of the system that identifies the unwanted events, find a feedback control strategy to minimize the number of these events.  We propose a two step solution to this problem. The first step is the identification of  the controllable causes of the unwanted events in the form of past time signal temporal logic (STL) formulas. The second step is the synthesis of a feedback controller that avoids the satisfaction of the cause formula found in the first step. In addition, we iteratively apply these two steps, namely cause identification and controller synthesis, to refine the controller and further reduce the number of these events.

Synthesis of STL formulas from a dataset has been studied in different forms including generation of a classifier to decide whether a trace exhibit a desired (or undesired) behavior~\cite{Bombara:2016,Kong:Inference:2014}, identification of high level system behaviors from traces~\cite{miningjournal,Jha2019}, and synthesis of monitoring rules~\cite{codit2018,acc2019}. In~\cite{miningjournal}, the goal is to synthesize parameters for a given template formula. For example, for the template formula $G_{[0,\infty]} x < p$, a valid valuation of the parameter $p$ gives an upper bound on the value of variable $x$ along all the traces in the dataset. On the other hand, in~\cite{Bombara:2016,Kong:Inference:2014,Jha2019,codit2018,acc2019,Bartocci2014}, both the formula structure (e.g. a template formula), and the parameters of it are found simultaneously. These works differ based on the structure of the dataset. In~\cite{miningjournal,Jha2019} only positive examples, whereas in~\cite{Bombara:2016,Kong:Inference:2014,Bartocci2014}, both positive and negative examples are considered. A common aspect of these works~\cite{Bombara:2016,Kong:Inference:2014,miningjournal,Jha2019,Bartocci2014} is that traces are considered as a whole (e.g. single label for a trace). Whereas, similar to this work, in~\cite{codit2018,acc2019}, each time point has a label, and we build on the results from~\cite{codit2018,acc2019}.

In the cause identification step, we generate a past time STL formula in a particular structure so that a controller can prevent the satisfaction of it. The formula represents a set of causes that can yield to the unwanted events in a human interpretable way. Each cause has a control part over the control variables and a general part 
that can be any STL formula over the system's state and control variables. 
The formula maps to a feedback controller that generates a control input violating each cause formula. 
 As opposed to the abstraction based~\cite{belta2017} and mixed integer linear program based controllers for STL specifications~\cite{Raman:2015:RSS,7447084}, the proposed controller does not require any offline computation and the system dynamics are not explicitly considered. At each time step, it simply checks the trace against a past time STL formula to generate a control input by exploiting the particular structure of the formula.

The main contribution of this work is the novel dynamics-independent framework for synthesis of feedback controllers to avoid unwanted events. This framework can be applied to any discrete time system with a finite control set. However, correctness is not guaranteed, i.e., the unwanted events may occur on the traces of the closed loop system. For this reason, we only aim at reducing the number of these events. The reduction rate depends on the causality relation between the control inputs and the unwanted events. We show on an example that the proposed controller can reduce the number of unwanted events significantly when the formal control problem is infeasible.


The paper is structured as follows. Sec.~\ref{sec:pre} presents necessary preliminaries. Sec.~\ref{sec:problem} formulates the cause identification and the controller synthesis problems. Sec.~\ref{sec:formula} presents our solution for the cause identification problem.  Sec.~\ref{sec:synthesis} describes the controller synthesis approach. Sec.~\ref{sec:conc} concludes the paper with future research directions.

\section{Preliminary Information}
\label{sec:pre}



\subsection{System Definition}

Consider the discrete time control system
\begin{equation}\label{eq:system}
x_{k+1} = f(x_k, u_k, w_k)
\end{equation}
where $x_k = [x_k^{0}, \ldots, x_k^{n-1}] \in \mathbb{X} \subset \mathbb{R}^n$ is the state of the system, $u_k = [u_k^{0}, \ldots, u_k^{m-1}] \in \mathbb{U} \subset \mathbb{R}^m$ is the control input and $w_k \in \mathbb{W} \subset \mathbb{R}^l$ is the noise at time step $k$. Each control input takes values from a finite set, i.e., $\mathbb{U}^i = \{c^{(i,1)},\ldots, c^{(i,M_i)}\} \subset \mathbb{R}, i=0,\ldots,m-1$ and $\mathbb{U} = \mathbb{U}^0 \times \ldots \times \mathbb{U}^{m-1}$. A finite length trajectory of system~\eqref{eq:system} is denoted by \[\mathbf{x} = (x_0, u_0), \ldots, (x_{N}, u_{N}),\] where for each $k=0,\ldots,N-1$, $x_{k+1} = f(x_k, u_k, w_k)$ for some $w_k \in \mathbb{W}$. The label of a trajectory  $\mathbf{x}$ is a binary sequence of the same length
\begin{equation}\label{eq:labels}
\mathbf{l} = l_0, \ldots, l_N
\end{equation}
where $l_k \in \{0,1\}$. The positive label, $l_k = 1$, indicates that a bad event occurred at time step $k$. The label sequence can be generated by a function over the state and system traces, e.g. $g: \mathbb{X} \times \mathbb{U} \to \{0,1\}$, and $l_k = g(x_k, u_k)$. A set of labeled traces of system~\eqref{eq:system} is denoted as 

\begin{equation}\label{eq:dataset}
\mathcal{D} = \{(\mathbf{x} _i, \mathbf{l} _i)\}_{i=1, \ldots, D}.
\end{equation}

\begin{example}\label{ex:trafficsystem}
We use a traffic system composed of 5 links and 2 traffic signals shown in Fig.~\ref{fig:traffic} as a running example. Each link of the system is modeled as a finite queue, which results in a piecewise-affine system. The details of the model can be found in~\cite{coogan2016traffic}. In this model, the system state captures the number of vehicles on each link and the modes of the traffic signals are control inputs. For the considered system, the capacity of the horizontal and vertical links are $40$ and $20$, respectively, i.e. $x^i \in [0,40]$ for $i \in \{0,1,2\} $, and $x^i \in [0,20]$ for $i \in \{3,4\} $.  The control inputs are defined as $\mathbb{U}^1 = \mathbb{U}^2 = \{0,1\}$, where $0$ denotes horizontal actuation and $1$ denotes vertical actuation. The system parameters are defined as follows (see~\cite{coogan2016traffic} for system dynamics): $c_i = 20$ for $i \in \{0,1,2\} $ and $c_i = 10$ for $i \in \{3,4\} $ (saturation flow), $d_i = 5$,  for $i \in \{0,3,4\}$ (exogenous flow bound), $\alpha^{0}_{0,1} = \alpha^{0}_{1,2} = 1, \alpha^{1}_{3,1} = \alpha^{1}_{4,2} = 1$ (supply ratios), $\beta_{0,1} = \beta_{1,2}=\beta_{2,3} = 0.75$, $\beta_{3,1} = \beta_{4,1} = 0.3$ (turn ratios).  

We simulated the system from random initial conditions and picked control values from $\mathbb{U}$ randomly at each time step to generate a dataset. In addition, we used the following formula to label the system traces:
\begin{equation}\label{eq:violationformula}
x^0 < 30 \wedge x^1 < 30  \wedge  x^2 < 30 \wedge x^3 < 15 \wedge x^4 < 15 
\end{equation}
Time points in which this formula is violated are labeled with $1$, all the other time points are labeled with $0$. Essentially, label 1 is generated whenever a link has more vehicles than $75\%$ of its capacity.We generate a dataset  $\mathcal{D}$ with 20 labeled traces. Each trace has length $100$. Out of $2000$ data points, $911$ of them were labelled with 1, which indicates that the system is congested $46\%$ of the time.
\end{example}

\begin{figure}[h]
\centering
\includegraphics[width=7cm]{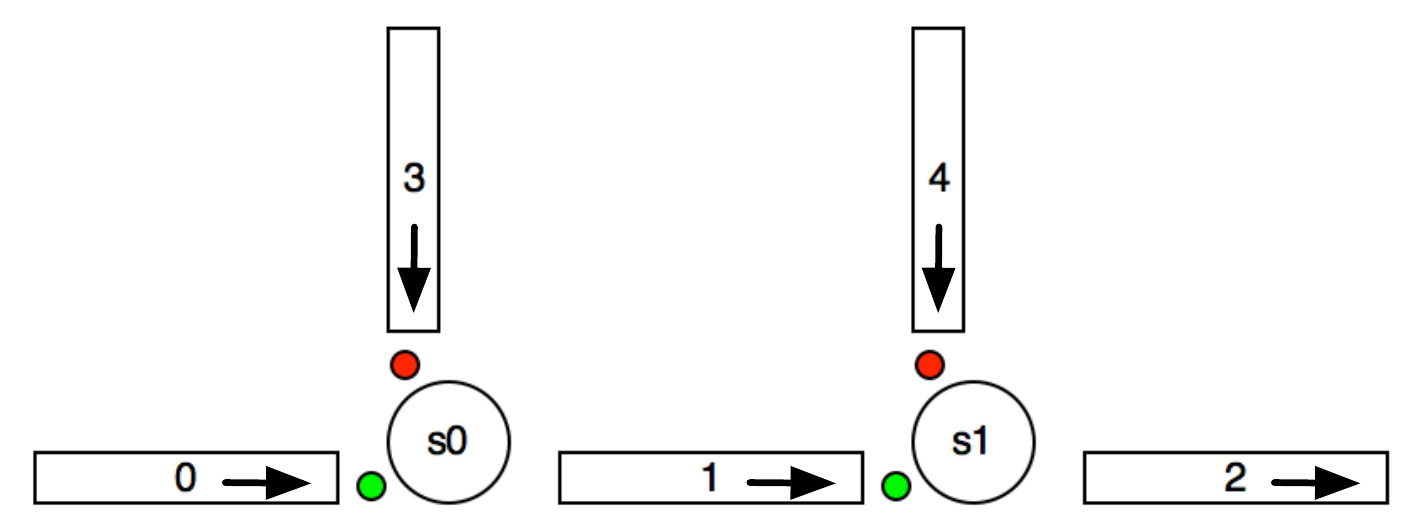}
\caption{A traffic network composed of 5 links and 2 signals. The flow directions are shown on links.}\label{fig:traffic}
\end{figure}

\subsection{Binary Classification}\label{sec:classification}

The goal in the binary classification task is to assign a binary label to each data point. The success of a classifier over a labeled dataset is computed according to the four basic categories of the classifier results and original labels: true positives, false positives, true negatives and false negatives.  

A binary classifier $P$ for traces of system~\eqref{eq:system} with length $M+1$ is defined as 
\begin{equation}\label{eq:classifier}
P: (\mathbb{X} \times \mathbb{U})^{M+1} \to \{0,1\}.
\end{equation}

For such a classifier, the computation of the number of true positives ($tp$) and false positives ($fp$) over a dataset $\mathcal{D}$~\eqref{eq:dataset} is shown in~\eqref{eq:tpfp}. The computation for the number of true negatives ($tn$) and false negatives ($fn$) are similar, and omitted for brevity. 
{\small{
\begin{align}\label{eq:tpfp}
tp &= \sum_{(\mathbf{x}, \mathbf{l}) \in \mathcal{D}} \sum_{i=M}^N \begin{cases} 1 \text{ if }P((x_{i-M}, u_{i-M}), \ldots, (x_{i}, u_{i})) == 1 \\ \quad\quad\quad\quad \text{ and } l_i == 1 \\ 0  \text{ otherwise }  \end{cases} \\
fp &= \sum_{(\mathbf{x}, \mathbf{l}) \in \mathcal{D}} \sum_{i=M}^N \begin{cases} 1 \text{ if }P((x_{i-M}, u_{i-M}), \ldots, (x_{i}, u_{i})) == 1 \\ \quad\quad\quad\quad \text{ and } l_i == 0 \\ 0  \text{ otherwise }  \end{cases} \nonumber
\end{align}
}}

The ratio of correctly classified instances ( $(tp+tn)/(tp+tn+fp+fn)$), precision ($ tp/(tp+fp)$), and recall ($ tp/(tp+fn)$) are commonly used success measures for classifiers. Another widely used metric is $F_\beta$ score, which combines precision and recall via a weight parameter $\beta$:
\begin{equation} 
F_{\beta} = \frac{(1 + {\beta}^2) * precision * recall }{( {\beta}^2 * precision ) + recall }
\end{equation}
$F_{\beta}$ takes values between $0$ (worst) and $1$ (best). Increasing $\beta$ puts more emphasis on recall (and false negatives), on the other hand , decreasing $\beta$  puts more emphasis on precision.



\subsection{Signal Temporal Logic}

We express signal specifications as ptSTL formulas~\cite{Asarin2012}. The syntax of ptSTL over the variables of system~\eqref{eq:system} is inductively defined as:\footnote{As the considered controls take values from finite sets, only equality constraints are used for controls.} 
\begin{align}\label{eq:ptstl}
\phi = & true \mid  x^i \sim c \mid u^i = c \mid \neg \phi \mid \phi_1 \wedge \phi_2 \mid  \phi_1 \since_{[a,b]} \phi_2  
\end{align}
where $true$ stands for the Boolean constant true, $x^i$ is a state variable, $u^i$ is a control variable, $\sim \in \{>, <\}$, $c \in \mathbb{R}$ is a constant, $[a, b] \subset \mathbb{R}$ represents a time interval with $a \leq b$, and $\since$ is ptSTL operator representing  \textit{Since}. 
The semantics of a ptSTL formula is defined over a signal for a given time point. The notation $(\mathbf{x},k) \models \phi$ is used to denote that the signal $\mathbf{x}$ satisfies $\phi$ at time $k$, which is defined as follows:

{\small{
\begin{align*}
& (\mathbf{x},k) \models x^i  \sim c   \iff   x^i_k \sim c    \text{ where } \sim \in \{>, <\}  \\
& (\mathbf{x},k) \models u^i  = c   \iff   u^i_k = c  \\
& (\mathbf{x},k) \models \neg \phi  \iff   (\mathbf{x},k) \not \models  \phi \\
& (\mathbf{x},k) \models  \phi_1 \wedge \phi_2  \iff   (\mathbf{x},k) \models\phi_1   \text{ and }  (\mathbf{x},k) \models  \phi_2\\
& (\mathbf{x},k) \models \phi_1 \since_{[a,b]} \phi_2 \iff  \text{there exists}  j \in ([k-b, k-a] \cap [0, k])  \\
& \text{ such that } (\mathbf{x},j)\models \phi_2 \text{ and for each } l \in [j, k], (\mathbf{x},l) \models \phi_1
\end{align*} 
}}
The \textit{Previously}|(Past Eventually)| ($\pastF$), \textit{Always}|(Past Globally)| ($\pastG$) operators are defined for notational convenience as they are special cases of $\since$ operator:
\[ \pastF_{[a,b]} \phi := true \since_{[a,b]} \phi \quad \pastG_{[a,b]} \phi := \neg \pastF_{[a,b]} \neg \phi  \]

\newtheorem{defn}{\textbf{Definition}} 
\begin{defn}[Parametric ptSTL Formula]\label{defn:definition1} Parametric ptSTL is an extension of STL that allows to represent some of the numeric constants with parameters \cite{Asarin2012}. Given a parametric ptSTL formula $\phi$ and a parameter valuation $v$, a ptSTL formula $\phi(v)$ is obtained by replacing parameters in $\phi$ from the corresponding constants in $v$, e.g. for formula $\phi= \pastF_{[0,b]}(x>c)$, and valuation $v: b=3,c=5$, $\phi(v) $ is $\pastF_{[0,3]}(x>5)$.
\end{defn}
\begin{defn}[Length of a formula]
Informally, the length of a formula gives the oldest time point that is required to evaluate the formula. The length is denoted with $\mathfrak{l}(\phi)$ and computed as follows:
\begin{align*}
&\mathfrak{l}(true)=\mathfrak{l}(x^i \sim c)=\mathfrak{l}(u^i=c)=0\\
& \mathfrak{l}( \neg \phi ) = \mathfrak{l}(  \phi ) , \quad  \mathfrak{l}( \phi_1 \wedge \phi_2 ) = \max( \mathfrak{l}( \phi_1), \mathfrak{l}( \phi_2 ) )\\
& \mathfrak{l}( \phi_1 \since_{[a,b]} \phi_2 ) = b + \max(\mathfrak{l}( \phi_1), \mathfrak{l}( \phi_2 ))
\end{align*}
\end{defn}

\begin{defn}[Operator count]
The \textit{operator count of a formula} is the total number of Boolean and temporal operators that appear in the formula. 
\end{defn}

A ptSTL formula $\phi$ can be used as classifier as given in ~\eqref{eq:classifier} for the traces of system~\eqref{eq:system}. In particular, $\mathfrak{l}(\phi)$ is the length of the trace fragment, and the classification result is 1 when the considered fragment satisfies the formula:
\[  (( (x_{i-\mathfrak{l}(\phi)+ 1} , u_{i-\mathfrak{l}(\phi)+ 1}), \ldots, (x_{i}, u_{i})   ), \mathfrak{l}(\phi) )  \models \phi    \] 

It is straightforward to derive the success measures, e.g. recall, precision, $F_\beta$-score, for a ptSTL formula $\phi$ over a dataset $\mathcal{D}$. In the remainder of the paper, $F_{\beta}(\phi(v), \mathcal{D})$ is used to denote the $F_\beta$ score  of the parametric formula $\phi$ with valuation $v$ over the dataset $\mathcal{D}$. The dataset is omitted (e.g. $F_{\beta}(\phi(v))$) when it is clear from the context.


\section{Problem Formulation}\label{sec:problem}

In this work, our goal is to design a feedback controller for system~\eqref{eq:system} to reduce the occurence of unwanted events during the execution of the system. We propose a two step solution for this problem. In the first step, we identify the \textit{controllable causes} of the unwanted events over a labeled dataset of system traces. To generate the dataset, we simulate the system with a nominal (or random) controller and label the traces. In the second step, we synthesize a controller that avoids the causes found in the first step. 

We propose to identify the controllable causes of the unwanted events in the form of ptSTL formulas over the system's state and control variables. In particular, the generated formula will represent the unwanted events (label 1). In addition, it will be in a particular structure, i.e, controllable, so that the satisfaction of the formula can be avoided by a controller. Thus the number of the unwanted events can be reduced by the controller. We call the formula found in the first step as \textit{the cause formula}. The formula synthesis problem is formalized in Prob.~\ref{prob:formula}.

\begin{problem}[Finding a cause formula]\label{prob:formula} Given a control system~\eqref{eq:system}, a set of its labeled traces $\mathcal{D}$~\eqref{eq:dataset}, find a ptSTL formula in the following form 
\begin{equation}\label{eq:formulashape}
\Psi  \coloneqq \Phi_1 \vee \ldots \vee \Phi_p,  
\end{equation} 
\[where\]
\begin{equation}\label{eq:formulainside}
  \quad \Phi_i \coloneqq (\pastG_{[1,b_i]} u^j=c_i) \wedge ( \pastF_{[1,1]} \phi_{i} ) , 
  \end{equation}
  and $\phi_i$ is any ptSTL formula over $\{ x^0, \ldots, x^{n-1}\} \cup \{u^0, \ldots, u^{m-1}\}$ 
  such that the valuation of the formula along the traces mimics the labels, i.e. for any $(\mathbf{x}, \mathbf{l}) \in \mathcal{D}$,  whenever $l_k = 1$, $(\mathbf{x}, k)\models \Psi$. 
  \end{problem}
  
The combined cause formula~\eqref{eq:formulashape} is a disjunction of cause formulas in the form of~\eqref{eq:formulainside} that represents a controllable cause. The controllable cause identifies the unwanted events in the previous time step. The first part  ($\pastG_{[1,b]} u^j=c_i$)  of the formula can be violated with  the choice of a suitable control input in the previous time step, and the second part ($ \pastF_{[1,1]} \phi_{i}$) captures the cause over the state and the control variables. 

 A formula that is satisfied at all of the time points with label $1$ (i.e. has no \textit{false negatives}) and that is violated at all other time points (i.e. has no \textit{false positives}) is an ideal formula. Due to the particular labeling process and the structure of the formula~\eqref{eq:formulashape} considered in Prob.~\ref{prob:formula}, it might not be possible to find such an ideal formula. In our solution, we aim at generating the \textit{best} formula representing the given dataset $\mathcal{D}$.  Any success measure for the binary classification task can be used to determine how good the formula is with respect to the given dataset. We use $F_\beta$-score to balance the false negatives and false positives. To generate the formula, we first sort all parametric formulas in the given form~\eqref{eq:formulainside} according to their operator counts. Then, iteratively perform parameter optimization and generate combined cause formulas as in~\eqref{eq:formulashape}. 
   

 In the second step of our solution, we synthesize a controller that avoids \textit{the combined cause formula} $\Psi$ found at the first step.

 \begin{problem}[Controller synthesis]\label{prob:controllersynthesis} Given a control system~\eqref{eq:system} and a ptSTL formula $\Psi$ of the form~\eqref{eq:formulashape}, generate a finite memory feedback controller $\mathcal{U}: (\mathbb{X} \times \mathbb{U})^K \to \mathbb{U}$ such that the trajectory of the closed loop system violates $\Psi$ at each time step.
\end{problem}

The proposed controller as a solution to Prob.~\ref{prob:controllersynthesis}  evaluates the trace of the system against the general and control parts of each cause formula, and produce the control inputs accordingly. In particular, at time step $k$, a control input $ u \in \mathbb{U}$ that violates time-shifted formula $(\pastG_{[0,b-1]} u^j=c_i) \wedge \phi_{i} $ for each subformula $\Phi_i$ is generated to violate $\Psi$ at the next time step. 
As these represent the unwanted events in the given dataset, the generated control strategy is expected to reduce the number of these events. Essentially, the stronger the causality between the control inputs and the unwanted events is, the better our solution is going to be, i.e., the larger number of unwanted events will be avoided by the synthesized controller.


    \begin{example}\label{ex:examplifyingformulas} For the traffic system shown in Ex.~\ref{ex:trafficsystem}, a formula $\phi$ as in~\eqref{eq:formulainside} represents the continuous actuation of a road for $b-1$ time units while blocking the neighbouring road, and the general part captures other properties in the system. For example, blocking $x^3$ when there is more than $10$ vehicles would lead to satisfaction of~\eqref{eq:violationformula} at the next time step, which is captured in the following cause formula:
\[ \Phi = (\pastG_{[1,1]} u^0=0) \wedge ( \pastF_{[1,1]} x^3  > 10 )  \]
    \end{example}

\section{Finding Controllable Causes}\label{sec:formula}


In this section, we present our solution for Prob.~\ref{prob:formula}. First, we define the set of all parametric cause formulas~\eqref{eq:formulainside} with a given operator count limit. Then, we use an iterative formula synthesis approach to generate $\Psi$~\eqref{eq:formulashape}. In particular, starting from $\Psi=false$, at each iteration, for each parametric cause formula $\Phi$, we synthesize parameters $v$ that optimize the valuation of the combined formula ($\Psi \vee \Phi(v)$), and choose the optimal one and concatenate it to $\Psi$ via disjunction. The iteration continues until a given cause or operator limit is reached or the improvement on the valuation is insufficient. 
We use $F_\beta$ score for formula evaluation to balance the errors (e.g. false positives and false negatives).  


A parametric ptSTL formula with $r$ operators is denoted as $\phi_{(oc=r)}$. Such a formula over the state and input variables of system~\eqref{eq:system} is recursively defined as follows:
\begin{align}
\phi_{(oc=r)} :=	& \phi_{(oc=r_1)} \vee \phi_{(oc=r_2)} \mid  \phi_{(oc=r_1)} \wedge \phi_{(oc=r_2)} \mid \nonumber \\ 
& \neg \phi_{(oc=r-1)}  \mid \pastF_{[a,b]} \phi_{(oc=r-1)} \mid  \pastG_{[a,b]}  \phi_{(oc=r-1)} \mid  \nonumber \\ & \phi_{(oc=r_1)} \since_{[a,b]} \phi_{(oc=r_2)}   \label{eq:parametric} \\
\phi_{(oc=0)} & := x^i < c \mid x^i > c \mid u^i = c \mid true \label{eq:parametric0}
\end{align} 
where $a,b$ and $c$ are parameters, and $r_1+r_2=r-1$. We use $ \Phi_{(oc=r)} $ to denote a parametric subformula~\eqref{eq:formulainside} with $r$ operators in the general formula part, which is defined as:
 \begin{align}\label{eq:parametricsubformula}
 \Phi_{(oc=r)} :=  (\pastG_{[1,b]} u^j=c) \wedge \pastF_{[1,1]} \phi_{(oc=r)},
 \end{align}

Note that the formula $(\pastG_{[1,b]} u^j=c ) \wedge (\pastF_{[1,1]} true)$ with $0$ operators in the general formula part is equivalent to $(\pastG_{[1,b]} u^j=c)$. This formula represents the case when the unwanted event occurred exclusively due to the choice of the control input.

 \begin{algorithm}[h]
\caption{Formula Search($\underline{oc}$ , $\overline{oc}$,$\overline{p}$, $\underline{val}$, $\mathcal{D}$, $\mathcal{P}$)}\label{alg:formulasearch}
\begin{algorithmic}[1]
\Require $\underline{oc}$: lower bound for the operator count, $\overline{oc}$: upper bound for the operator count, $\overline{p}$: upper bound for the number of subformulas, $\underline{val}$: lower bound for the improvement on the valuation, $\mathcal{D}$: a dataset as in~\eqref{eq:dataset}, $\mathcal{P} = \{ \mathcal{P}(\alpha) \mid \alpha \in \{a, b\} \cup \{c_{x,i} \mid i=0,\ldots,n-1\} \cup \{c_{u,i} \in i=0,\ldots,m-1\}\} $: a finite domain $\mathcal{P}(\alpha) $ for each parameter type. 

$\bigvee(\mathcal{L})$ is the formula obtained by concatenating all formulas in $\mathcal{L}$ with $\vee$, e.g.,  $\bigvee \left(\{\Phi_1, \Phi_2, \Phi_3\}\right) :=\Phi_1 \vee\Phi_2 \vee \Phi_3$.
\State $oc, {oc}_{last} \gets \underline{oc},$ $\quad \quad FS = \{false\}$ 
\State Compute $PF^{= oc}$  (set of parametric formulas as in~\eqref{eq:parametricsubformula})\label{line:formulas}
\While {$oc \leq \overline{oc}$ \textbf{and} ${oc}_{last} - oc < oc_{lim}$  \textbf{and} $\mid FS \mid \leq \overline{p}$}\label{line:whilestart}
	\ForAll {$\Phi \in PF^{= \overline{oc}}$}\label{line:forstart}
		\State  $\Phi(v),  F_{\beta}(\bigvee(FS)\vee\Phi(v)) = ParameterOptimization(\mathcal{D}, \bigvee(FS) \vee \Phi, \mathcal{P})$\label{line:parameteroptimization}
	\EndFor\label{line:loopend}
	\State $\Phi(v)^*= \Phi(v)$ with the best $F_{\beta}(\bigvee(FS)\vee\Phi(v))$ \label{line:bestformula}
	\If  {$F_{\beta}(\bigvee(FS) \vee \Phi(v)^*) - F_{\beta}(\bigvee(FS)) \leq \underline{val}$ }  \label{line:valuationlimitcheck}
		\State $oc = oc + 1$ \label{line:incrementoc}
		\State Compute $PF^{= oc}$
	\Else \label{line:else}
		\State $FS = FS \cup \{\Phi(v)^*\} $ \label{line:addcause}
		\State ${oc}_{last} \gets oc$
	\EndIf
\EndWhile\label{line:endwhile}
\State \Return $\bigvee(FS)$\label{line:returnline}
\end{algorithmic}
\end{algorithm}

The proposed formula synthesis method is summarized in Alg.~\ref{alg:formulasearch}. In this algorithm, first, the current operator count $oc$ and the last used operator count ${oc}_{last}$ are set to the minimum operator count $\underline{oc}$. Then, the set of all parametric cause formulas $PF^{={oc}}$ with $oc$ operators (as defined in~\eqref{eq:parametricsubformula}) is computed~(line~\ref{line:formulas}). In the main loop from line~\ref{line:whilestart} through line~\ref{line:endwhile}, parameter optimization is performed for parametric cause formulas to generate the combined cause formula in an iterative fashion. In particular, initially, the set of optimized cause formulas, $FS$, includes only formula $false$ and parameter optimization is performed for each parametric formula in $PF^{={\underline{oc}}}$ (lines~\ref{line:forstart}-\ref{line:loopend}). The optimum among these, $\Phi(v)^*$, is selected. If the difference of the evaluations of $\Phi(v)^*$ and formula $false$ is greater than the given bound $\underline{val}$, $\Phi(v)^*$ is added to the set of optimized cause formulas (line~\ref{line:addcause}). Otherwise, considering that the valuation of the best cause formula with operator count $oc$ is below the limit, the operator count is incremented by one (line~\ref{line:incrementoc}) and $PF^{= oc}$ is computed. In the subsequent iterations of the main loop, for the parametric cause formulas $\Phi \in PF^{= oc}$, the parameter optimization is performed on the combined cause formula $\bigvee(FS) \vee \Phi$ (line~\ref{line:parameteroptimization}, $\bigvee(FS)$  is not parametric), and, again, the valuation of the combined formula, $F_{\beta}(\bigvee(FS) \vee \Phi(v))$, is used in formula selection (line~\ref{line:bestformula}). Thus, in both cases, the improvement achieved with the new cause formula $\Phi$~\eqref{eq:parametricsubformula} is considered. At each iteration of the main loop, either a new formula is added to $FS$ or the operator count is incremented. The loop terminates when 1) the operator count exceeds the given bound $\overline{oc}$, or 2) a new cause was not added for the last $oc_{lim}$ different operator counts, or 3) the given cause limit $\overline{p}$ is reached. $FS$ is the set of all optimized cause formulas, and the disjunction of these is returned (line~\ref{line:returnline}).

The \textit{ParameterOptimization} method takes a dataset $\mathcal{D}$~\eqref{eq:dataset}, a parametric ptSTL formula $\Phi$, and a domain for each parameter type $\mathcal{P}$ as input. It iterates through all possible valuations $v$ for the parameters that appear in $\Phi$. In particular, each parameter in $\Phi$ is considered to be different and the optimization method iterates through the product of the corresponding parameter domains, e.g., for formula $\pastG_{[1,b]} (u^j = c) \wedge \pastF_{[1,1]} \pastG_{[a,b]} (u^j = c)$ iterate through $ \mathcal{P}(b) \times \mathcal{P}(c_{u,j}) \times  \mathcal{P}(a) \times \mathcal{P}(b) \times  \mathcal{P}(c_{u,j})$. 
For each valuation $v$, it computes the formula $\Phi(v)$ as in Defn.~\ref{defn:definition1} and its fitness $F_{\beta}(\Phi(v), \mathcal{D})$. Finally, it returns the valuation with the highest score, e.g., $\text{arg}\max\{ F_{\beta}(\Phi(v), \mathcal{D}) \mid v \text{ is a valuation for } \Phi \}$.

\begin{example}\label{ex:formulasynthesis}
We apply the parameter synthesis algorithm on the dataset defined in Ex.~\ref{ex:trafficsystem} for the traffic system shown in Fig.~\ref{fig:traffic} ($n=5$, $m=2$). The parameter domains are defined as follows: $\mathcal{P}(c_{x,i}) = \{10, 15, 20, 25, 30, 35\}$ for $i=0,1,2$, $\mathcal{P}(c_{x,i}) = \{5,10,15\}$ for $i=3,4$, $\mathcal{P}(c_{u,i}) = \{0,1\}$ for $i=0,1$,  $\mathcal{P}(a) = \{0, 1, 2, 3, 4, 5\}$, and $\mathcal{P}(b) =  \{0, 1, 2, 3, 4, 5\}$. We run Alg.~\ref{alg:formulasearch} with parameters $\underline{oc}=0$, $\overline{oc}=\infty$, $oc_{lim}=2$, $\overline{p}=\infty$. Note that no upper bound is set on the number of cause formulas or on the number of operators. The algorithm only terminates when no formula from $PF^{oc}$ satisfies the valuation condition for the last 2 operator counts (line~\ref{line:valuationlimitcheck}). The resulting formulas for $\underline{val} = 0.1$ and $\underline{val} = 0.01$ are:
\begin{align*}
\Psi^{*}_{0.1}(v) &=  (\pastG_{[1,3]} u^0 = 0 ) \vee  (\pastG_{[1,3]} u^0 = 1 ) \\ 
\Psi^{*}_{0.01}(v) &=  (\pastG_{[1,3]} u^0 = 0 ) \vee  (\pastG_{[1,3]} u^0 = 1 ) \vee \\
& ( (\pastG_{[1,1]} u^0 = 0 ) \wedge (\pastF_{[1,1]} ( x^3 > 10 ) ) ) \vee \\
& ( (\pastG_{[1,1]} u^1 = 0 ) \wedge (\pastF_{[1,1]} ( x^4 > 10 ) ) )
\end{align*}
As expected, smaller valuation limit $\underline{val}$ resulted in longer formulas with higher fitness: $F_{0.3}(\Psi^{*}_{0.1}(v)) = 0.92$, $tp$ of $\Psi^{*}_{0.1}(v)$ is $427$, whereas  $F_{0.3}(\Psi^{*}_{0.01}(v)) = 0.99$, and $tp$ of $\Psi^{*}_{0.01}(v)$ is $785$. The $fp$ for both of the formulas is $0$ (see~\eqref{eq:tpfp}). Essentially, when $\underline{val} = 0.1$, only the formulas that improve the valuation significantly is added to the optimized formula set. The computation took $112$ and $260$ seconds for  $\Phi^{*}_{0.1}(v)$, and $\Phi^{*}_{0.01}(v)$, respectively, on a laptop with $4GB$ memory and $1.6$ GHz Intel Core $i5$ processor. Formula $\Phi^{*}_{0.1}(v)$ shows that blocking link $0$ or link $3$ for three consecutive time units will cause congestion (see~Fig.\ref{fig:traffic}). The first part of the formula $\Phi^{*}_{0.01}(v)$ found in the second case is the same as $\Phi^{*}_{0.1}(v) $. The remaining part indicates that if the link 3 or link 4 is blocked when there are more than $10$ vehicles on them, there will be congestion in the next time step. A sample system trace, the label of it and the valuation of formula $\Psi^{*}_{0.01}(v)$ along the trace are shown in Fig.~\ref{fig:simulation}. In addition, we validated $\Psi^{*}_{0.01}(v)$  on $10$ other randomly generated datasets as defined in Ex.~\ref{ex:trafficsystem} (using a random controller). Similar to this example, each dataset has $20$ traces with length $100$. The average number of positive labels in a dataset is  $939.7$ (congestion for  $47\%$ of the time in average). The average number of true positives and false positives for $\Psi^{*}_{0.01}(v)$ over these datasets are $803.9$ and $0.4$, respectively.  These experiments show that the formula correctly identifies the causes for $85\%$ of the unwanted events, and rarely marks a normal event as bad. The minimality of the false positive rates is a result of using $F_{0.3}$-score. \end{example}

\begin{figure}[h]
\centering
\includegraphics[width=8.5cm]{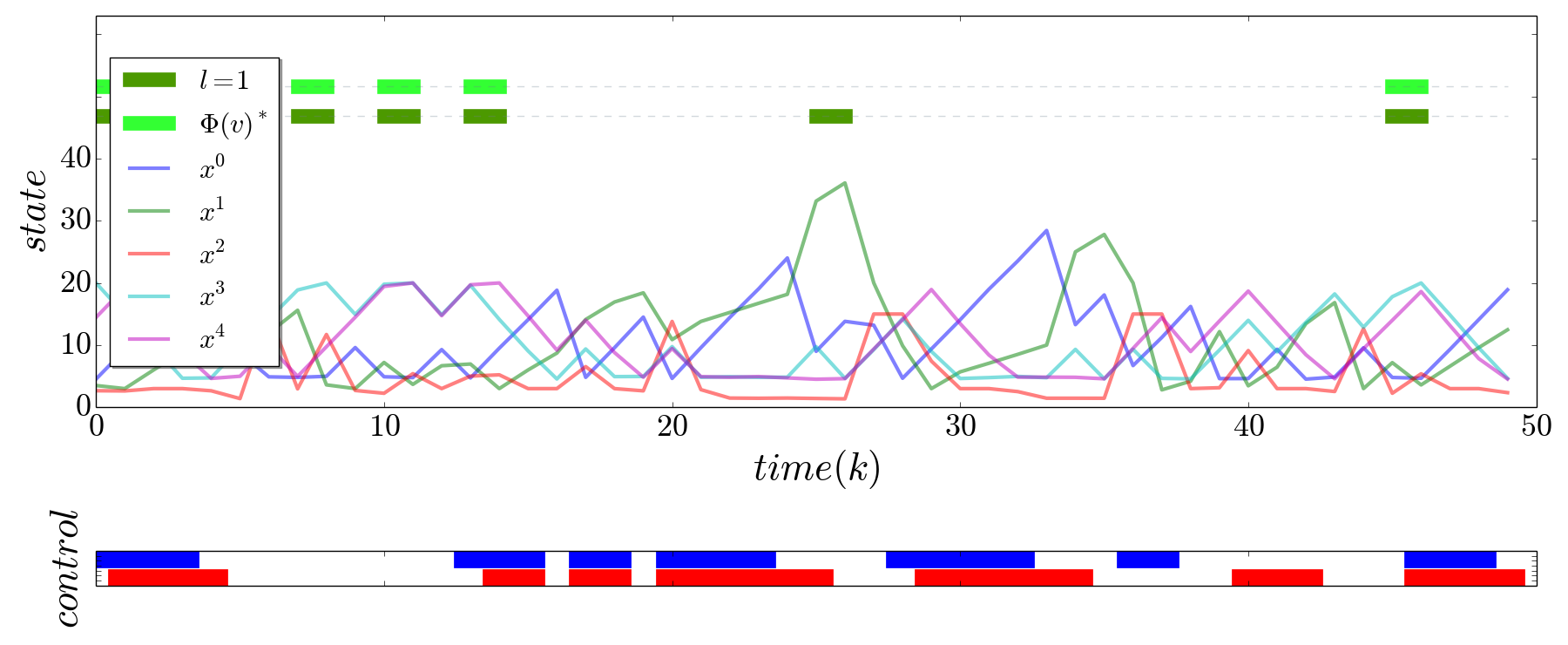}
\caption{A sample system trace, its label and valuation of $\Psi^{*}_{0.01}(v)$. In the bottom plot, the blue line shows that control $u^0$ is $1$ and the red shows that control $u^1$ is $1$ .}\label{fig:simulation}
\end{figure}

\textbf{Complexity.} The time complexity of Alg.~\ref{alg:formulasearch} is characterized by the operator count limit $\overline{oc}$, the cause limit $\overline{p}$ and the size of the parameter domain $s$, where $s = \max_{ \mathcal{P}(\alpha) \in \mathcal{P}} |\mathcal{P}(\alpha)|$. The number of parameters $\#(\phi)$ of a parametric ptSTL formula $\phi$ with $r$ operators is upper bounded by $3r + 1$, where the upper bound is attained when the formula has  $r$ nested $\since$ (since) operators. The complexity of the parameter optimization method is characterized by the number of fitness value computations, which is exponential in the number of parameters of the input formula $\phi$, $\#(\phi)$, and polynomial in $s$, e.g., $O(s^{\#(\phi)})$ due to the grid search. At each iteration, the parameter optimization is performed for each formula in $PF^{= {oc}}$. The number of parametric formulas with ${oc}$ operators, i.e., $| PF^{= {oc}} |$, is exponential in $oc$, and polynomial in $n$ (system dimension) and $m$ (control dimension). In particular, $| PF^{= {0}} | = m \times (2n+m+1)$ is the number of parametric formulas with 0 operators (see~\eqref{eq:parametric0}) and $| PF^{= {oc}} |$ is in $O( | PF^{= {0}} | ^{{oc}})$. The algorithm terminates with in $\overline{oc} + \overline{p}$ iterations. Even though the parameter optimization method is exponential in the number of parameters, due to the iterative nature of the algorithm, at most $3\overline{oc} + 3$ parameters are optimized simultaneously. 

In Alg.~\ref{alg:formulasearch}, a grid search is performed for each candidate cause formula via the \textit{ParameterOptimization} method. Alternative optimization methods such as formula synthesis via data mining~\cite{acc2019} can be employed to improve the computational efficiency. In addition, monotonicity based efficient parameter synthesis methods can be adapted to the considered dataset when a monotone optimization criteria, such as the number of true positives, is used instead of $F_\beta$ score~\cite{miningjournal}. As the particular parameter optimization method is not the focus of this work, we used a simple search algorithm.

\section{Controller Synthesis}\label{sec:synthesis}


In this section, we present our solution for Prob.~\ref{prob:controllersynthesis}. The proposed feedback controller generates a control input $u \in \mathbb{U}$ at each time step to avoid the satisfaction of the cause formula $\Psi=\Phi_1 \vee \ldots \vee \Phi_p$~\eqref{eq:formulashape} found in the first step. Note that, at time step $k$, $\Psi$ is violated if and only if each subformula $\Phi_i$ for $ i=1,\ldots,p$ is violated.

 \begin{algorithm}
\caption{Controller($\Psi$, $\mathbf{x}$, $x_k$)}\label{alg:controller}
\begin{algorithmic}[1]
\Require $\Psi=\Phi_1 \vee \ldots \vee \Phi_p$: the cause formula, each $\Phi_i = ( \pastG_{[1,b_i]} u^j=c_i) \wedge ( \pastF_{[1,1]} \phi_{i} )$~\eqref{eq:formulainside},
$\mathbf{x} = (x_{k-\mathfrak{l}(\Phi)+ 1} , u_{k-\mathfrak{l}(\Phi)+ 1}), \ldots, (x_{k-1}, u_{k-1})$: partial trace of length $\mathfrak{l}(\Phi) - 1$, $x_{k}$: state of the system at time~$k$
\Ensure $u_{k} \in \mathbb{U}$ and $(x_{k-\mathfrak{l}(\Phi)+ 1} , u_{k-\mathfrak{l}(\Phi)+ 1}), \ldots, (x_{k}, u_{k})$ violates $\Phi$ at time $k$. 
\State $\mathbb{U}^{cand} = \mathbb{U}$
\While{$\mathbb{U}^{cand} \neq \emptyset $}\label{line:randomControlS}
\State $u_k = Random(\mathbb{U}^{cand})$
\If{$ ((\mathbf{x}, (x_k, u_k) ),  k ) \models ( \pastG_{[0,b_i-1]} u^j=c_i) \wedge\phi_{i}  $  for some $\Phi_i$ from $\Psi$}
\State $\mathbb{U}^{cand} = \mathbb{U}^{cand} \setminus \{u_k\}$\label{line:removeControl}
\Else $\quad$ \Return $u_k$
\EndIf
\EndWhile\label{line:randomControlE}
\State \Return $Random(\mathbb{U})$ \label{line:randomControl}
\end{algorithmic}
\end{algorithm}

The controller is presented in Alg.~\ref{alg:controller}. It takes the cause formula $\Psi$ and the partial trace $\mathbf{x}$ of the system that is required to compute the satisfaction of $\Psi$ at time step $k$ as inputs. In line~\ref{line:randomControlS} through line~\ref{line:randomControlE}, the controller randomly picks a control input $u$ from $\mathbb{U}$, appends $(x_k, u)$ to the partial trace $\mathbf{x}$, then for each sub-formula $\Phi_i = (\pastG_{[1,b]} u^j=c) \wedge \pastF_{[1,1]} \phi$, the controller evaluates the formula shifted by a time step, $(\pastG_{[0,b-1]} u^j=c) \wedge \phi$, along the partial trace. $u$ is removed from the candidate set of controls (line ~\ref{line:removeControl}) if the shifted-formula is satisfied for a cause formula $\Phi_i$. Otherwise, $u$ guarantees that each cause formula, thus $\Psi$~\eqref{eq:formulashape}, will be violated at time step $k+1$, and the controller generates $u$. 

If the candidate control set becomes empty, the controller generates a control input randomly (line~\ref{line:randomControl}). In this case, the control input that satisfies the minimum number of cause formulas can be generated with additional bookkeeping.In the following proposition, we establish the sufficient conditions under which the controller guarantees that $\neg \Psi$ is satisfied along the system trace.



\begin{proposition}\label{prop:controller} Given a ptSTL formula $\Psi$~\eqref{eq:formulashape}, control system~\eqref{eq:system},the state $x_k$, a partial trace \[\mathbf{x} = (x_{k-\mathfrak{l}(\Phi)} , u_{k-\mathfrak{l}(\Phi)}), \ldots, (x_{k-1}, u_{k-1}),\] the control $u_k$ generated by Alg.~\ref{alg:controller} guarantees that $\Psi$ is violated at time step $k$ when $b_i \geq 2$ for each cause formula $\Phi_i$ and $| \mathbb{U}^j | \geq 2$ for each $j=0,\ldots,m-1$. 
\end{proposition}
\begin{proof} Let $\textbf{U}_k = \mathbb{U}^0 \setminus \{u_{k-1}^0\} \times \ldots \times \mathbb{U}^{m-1} \setminus \{u_{k-1}^{m-1}\}$ where $u_{k-1}= [ u_{k-1}^0 , \ldots , u_{k-1}^{m-1}]$ is the control input of time step $k-1$. Since $| \mathbb{U}^j | \geq 2$ for each $j$, $\textbf{U}_k \neq \emptyset$. Consider $u_k \in \textbf{U}_k$ and a shifted cause formula  $\Phi^s_i = \pastG_{[0,b_i-1]} u^j=c_i \wedge \phi_i$. Since $b_i-1 \geq 1$ and $u_k^j \neq u_{k-1}^j$, it holds that $((\mathbf{x}, (x_k, u_k) ), k)\not \models \Phi^s_i$. Note that $({(}\mathbf{x}, (x_k, u_k) {), k})\not \models \Phi^s_i$ implies that $({(}\mathbf{x}, (x_k, u_k),(x_{k+1}, u_{k+1}) ),k+1 )\not \models \Phi_i$ for any $x_{k+1} \in \mathbb{X}$ and $u_{k+1} \in \mathbb{U}$. Observe that $\textbf{U}_k \subset \mathbb{U}$, thus $u_k$ can be generated by Alg.~\ref{alg:controller}.
Consequently, formula $\Psi$ is violated at each time step since each cause formula $\Phi_i$ is violated.
\end{proof}

 

\begin{example}\label{ex:controller}
We simulated the traffic system from Ex.~\ref{ex:trafficsystem} in closed loop with the controller generated from formula $\Psi^{*}_{0.01}(v)$ (see Ex.~\ref{ex:formulasynthesis})as in Alg.~\ref{alg:controller}. The congestion percentage is dropped to $7.5\%$ from $46\%$ (in average). 
\end{example}

The developed framework identifies the controllable causes of the unwanted events from the labeled dataset, and reduces the occurrences of the labeled events via controller synthesis as illustrated in Ex.~\ref{ex:formulasynthesis} and Ex.~\ref{ex:controller}. An insufficient reduction on the unwanted events can be analyzed in three aspects when there is a high correlation between the unwanted events and the control inputs. 1) The complex causes can not be explained with limited number of operators ($\overline{oc}$). 2) The number of different causes yielding the labelled events is higher than the bound used in the formula synthesis ($\overline{p}$). 3) Avoiding formula $\Psi$ results in formation of new causes. For the first and the second one, the parameters of the formula synthesis algorithm can be increased. However, due to the time complexity of Alg.~\ref{alg:formulasearch}, increasing both of them might not be practical. The second and third issues can be solved via iterative applications of the formula generation and controller synthesis steps, which is presented in the next section.


\subsection{Controller Refinement}

The iterative algorithm to reduce the number of unwanted events via cause identification and controller synthesis is summarized in Alg.~\ref{alg:iterativeapproach}. In addition to the parameters of Alg.~\ref{alg:formulasearch}, Alg.~\ref{alg:iterativeapproach} takes as input the system $S$ and a function over the state and control variables of the system to mark the unwanted events. Initially, the combined cause $\Psi$ is set to $false$ and a labeled dataset is generated by simulating the system in closed loop with a random control generator. Then, iteratively, a formula $\Phi$ is generated from the dataset $\mathcal{D}$ with Alg.~\ref{alg:formulasearch}, it is added to the previously identified causes $\Psi$ via disjunction, a controller $\mathcal{U}$ is synthesized from the new cause  formula $\Psi$ with Alg.~\ref{alg:controller}, and finally a new dataset is constructed by simulating the system in closed loop with the new controller $\mathcal{U}$. This iterative process continues until the rate of the unwanted events drops below a predefined limit $B$. 


 \begin{algorithm}
\caption{ControllerSynthesis($S$, $g^{viol}$, $\overline{oc}$, $\underline{oc}$, $\overline{p}$, $\underline{val}$, $\mathcal{P}$)}\label{alg:iterativeapproach}
\begin{algorithmic}[1]
\Require $S$ is a control system as in~\eqref{eq:system}, $g^{viol}: \mathbb{R}^n \times \mathbb{U} \to \{0,1\}$ is function generating labels for the unwanted events, $\overline{oc}$, $\underline{oc}$, $\overline{p}$, $\underline{val}$, and $\mathcal{P}$ are as in Alg.~\ref{alg:formulasearch}
\Ensure $\mathcal{U} : (\mathbb{R}^n \times \mathbb{R}^m)^K \to\mathbb{U} $ is a feedback control strategy  minimizing $\sum_{i=0}^N  g^{viol} (x_i, u_i)$ along the traces $\mathbf{x}=(x_0, u_0), \ldots, (x_N, u_N)$ of the closed loop system. 
\State $\Psi = False$
\State Generate $\mathcal{D}$, simulate $S$ by randomly choosing controls and label the traces with $g^{viol}(\cdot,\cdot)$
\Repeat
\State $\Phi$ = FormulaSearch($\overline{oc}$, $\underline{oc}$, $\overline{p}$, $\underline{val}$, $\mathcal{D}$, $\mathcal{P}$).
\State $\Psi = \Psi \vee \Phi$
\State Define $\mathcal{U}: Controller(\Psi, \mathbf{x},x_k)$ as in Alg.~\ref{alg:controller}.
\State Generate $\mathcal{D}$: simulate $S$ in closed loop with $\mathcal{U}$ and label the traces with $g^{viol}$
\Until $\frac{\sum_{(\mathbf{x},\mathbf{l}) \in \mathcal{D}} \sum_{i=1,\ldots,N}  l_i } {N \times |\mathcal{D}| } \leq B$
\end{algorithmic}
\end{algorithm}

\begin{example}
We run the iterative synthesis algorithm (Alg.~\ref{alg:iterativeapproach}) for the traffic system from Ex.~\ref{ex:trafficsystem} with different values for $\underline{oc}$, $\overline{oc}$, $\overline{p}$, $\underline{val}$. In each experiment, the parameter domain given in Ex.~\ref{ex:formulasynthesis} is used. 
In all of the experiments, we set $B$ to $0$ so that the refinement loop terminates only when the congestion is avoided. The total number of labeled data points, $\mathcal{V}_i$, in dataset $\mathcal{D}$ for each iteration $i$ and  the total computation time are reported in Table~\ref{tab:iterative1}. The parameters used in each experiment are shown in Table~\ref{tab:iterativeParameters}. Each experiment is started with the same dataset that was used in the running example, e.g., $\mathcal{V}_0$ is the violation count in this dataset. In each case the proposed iterative controller synthesis algorithm is able to generate a control strategy that avoids the unwanted events. 

In the first two experiments, no limit is set for the number of operators and the number of causes as in Ex.~\ref{ex:formulasynthesis}. In the first experiment, due to the low valuation limit, detailed cause formulas with high operator counts are generated by Alg.~\ref{alg:formulasearch}. While a controller that avoids the congestion is synthesized in fewer iterations, due to the complexity of parameter synthesis, the computation takes longer. In the second experiment, the valuation limit is increased to $0.1$. As explained in Ex.~\ref{ex:formulasynthesis}, cause formulas with fewer number of operators are generated. The total computation time is reduced while the number of iterations is increased. The parameters used in the last experiment enforce formula synthesis algorithm to generate a single cause with operator count $0$, the general formula $\phi_i$ is in the form of~\eqref{eq:parametric0}. Note that $\underline{val}$ is ineffective in this case. The refinement algorithm is able to reduce the violation count to $0$ in seven iterations. The parameter optimization is performed for at most $3$ parameters ($oc=0$), which results in short computation time. 
\end{example}

\begin{table}[h!]
\begin{center}
 \begin{tabular}{|c c c c c c c c c c|}
 \hline
 $Ex \# $ & time & {$\mathcal{V}_0$}  & {$ \mathcal{V}_1 $} & {$\mathcal{V}_2 $}   & {$\mathcal{V}_3$} & {$\mathcal{V}_4$} & {$\mathcal{V}_5$} & {$\mathcal{V}_6$} & {$\mathcal{V}_7$} \\ [0.5ex] 
 \hline\hline
  \#1 & 630sec & {911} & {143} & {23} &{5} & {2} & {0} & - & -\\ 
 \hline
  \#2 & 320sec & 911 & 530 & 265 & 5 & 6 &  2 & 0 & -\\ 
 \hline
 \#3 & 115sec & 911 &617 & 149 & 80 & 70 & 48 & 10 & 0  \\
 \hline
\end{tabular}
\end{center}
\caption{Violation counts after each iteration of Alg.~\ref{alg:iterativeapproach}}\label{tab:iterative1}
\end{table}

\begin{table}[h!]
\begin{center}
 \begin{tabular}{|c c c c c|}
 \hline
 $Ex \# $ & $\underline{oc}$ & $\overline{oc}$  &  $\overline{p} $  & $\underline{val}$ \\ [0.5ex] 
 \hline\hline
  \#1 & 0 & $\infty$ & $\infty$ & 0.01 \\ 
 \hline
  \#2 & 0 & $\infty$ & $\infty$ & 0.1 \\ 
 \hline
 \#3 & 0 & 0  & 1 & 0 \\
 \hline
\end{tabular}
\end{center}
\caption{Parameters of Alg.~\ref{alg:iterativeapproach}}\label{tab:iterativeParameters}
\end{table}



\section{Case Study}\label{sec:synthesis}

\subsection{Congested Traffic System}

In this case study, we increase the exogenous flow limit of link $0$ to $10$ of the traffic system defined in Ex.~\ref{ex:trafficsystem}. All other parameters are kept the same. Due to the increase in the inflow, the average congestion is reached to $79\%$. In this example, we used the abstraction based synthesis approach~\cite{coogan2016traffic} to generate a control strategy to avoid congestion (defined in~\eqref{eq:violationformula}). The synthesis problem is infeasible even when the capacity of each link is partitioned into regions of length $1$ (i.e. the size of the abstraction is $40^3 \times 20^2$).

We run Alg.~\ref{alg:iterativeapproach} on this system with parameters $\overline{oc}=0$, $\underline{oc}=0$, $\overline{p}=1$, and $\underline{val}=0$.  For each iteration, the total number of labeled data points $\mathcal{V}_i$, the formula $\Phi_i$ generated by Alg.~\ref{alg:formulasearch} and its $tp$ and $fp$ count are shown in Table~\ref{tab:congested}.

\begin{table}[h!]
\begin{center}
 \begin{tabular}{|c| c | c | c| c| }
 \hline
 $i $ & {$\mathcal{V}_i$} & $\Phi_i$  & $tp$ & $fp$ \\ [0.5ex] 
 \hline\hline
  1 & 1591 & $ (\pastG_{[1,2]} u^0 = 1 ) \wedge (\pastF_{[1,1]} ( x^0 > 20 ) )  $ & 405 & 0 \\ 
 \hline
  2 & 1549 & $ (\pastG_{[1,2]} u^0 = 0 ) \wedge (\pastF_{[1,1]} ( x^3 > 10 ) )  $ & 576 & 0 \\ 
\hline
  3 & 1557 & $ (\pastG_{[1,2]} u^1 = 0 ) \wedge (\pastF_{[1,1]} ( u^1= 0 ) )  $ & 465 & 19 \\ 
\hline
  4 & 1763 & $ (\pastG_{[1,2]} u^1 = 1 ) \wedge (\pastF_{[1,1]} ( u^1 = 1 ) )  $ & 599 & 3 \\ 
\hline
  5 & 2 & {-}& {-} & {-} \\ 
\hline\end{tabular}
\end{center}
\caption{Congested traffic system: Violation counts and the optimized cause formulas for each iteration of Alg.~\ref{alg:iterativeapproach}.}\label{tab:congested}
\end{table}

The resulting formulas show that even though the valuation of the formula synthesized in Alg.~\ref{alg:formulasearch} is high (see $tp$ and $fp$), the number of violations in the resulting system may not decrease ({$\mathcal{V}_i$} in the next row). This is due to the formation of the new causes. For example, the controller that avoids $\Phi_1$ favors $u^0 = 0$, thus causes congestion on link $3$. $\Phi_1$ and $\Phi_2$ together can balance the control choice for the first signal $u^0$. The same situation occurs for the second signal $u^1$ with $\Phi_3$ and $\Phi_4$. In this example, the controller generated from $\Psi = \Phi_1 \vee \ldots \vee \Phi_4$ reduces the total congestion count to $2$ ($0.05\%$), and subsequent iterations do not reduce it further.

This example shows that the proposed iterative controller synthesis algorithm is able to reduce the number of unwanted events significantly when a control strategy guaranteeing the avoidance of these events does not exist.

\subsection{Region Avoidance}

In this case study, we consider a basic region/obstacle avoidance scenario for a robot that can move in a planar arena. The arena is represented as a two-dimensional grid of size $M_0 \times M_1$. The state of the robot is its location, i.e., $x_k \in [0,\ldots,M_0-1] \times [0,\ldots,M_1-1] $, and the control set is $\mathbb{U} = \{N,E,S,W\}$. The robot can move according to the control input, but it can not leave the arena: 
\begin{align}\label{eq:toy}
[ x^0_{k+1}, x^1_{k+1} ] = \begin{cases}   [ \min(0,x^0_{k} - 1), x^1_{k} ] \text{ if } u = N\\
								[x^0_{k}, \max(x^1_{k} + 1, M_1-1)]\text{ if } u = E\\
							 [ \max(x^0_{k} + 1, M_0-1), x^1_{k} ] \text{ if } u = S \\
							 [ x^0_{k} , \min(0,x^1_{k}-1) ] \text{ if } u = W
							 \end{cases}
\end{align}

%

\begin{figure}[h]
\centering
\begin{subfigure}[b]{0.21\textwidth}
\includegraphics[width=\textwidth]{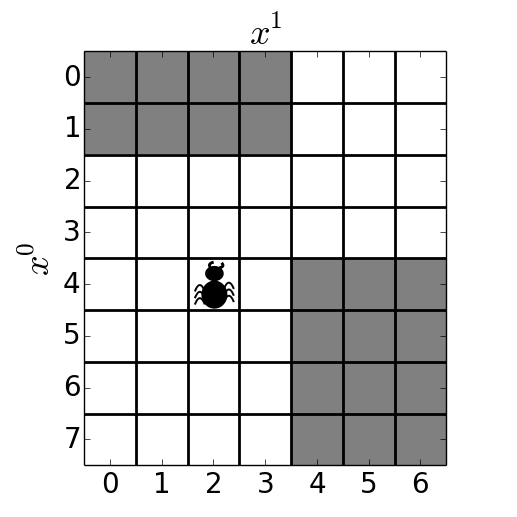}\caption{}
        \label{fig:arena}
    \end{subfigure}
    \begin{subfigure}[b]{0.24\textwidth}
\includegraphics[width=\textwidth]{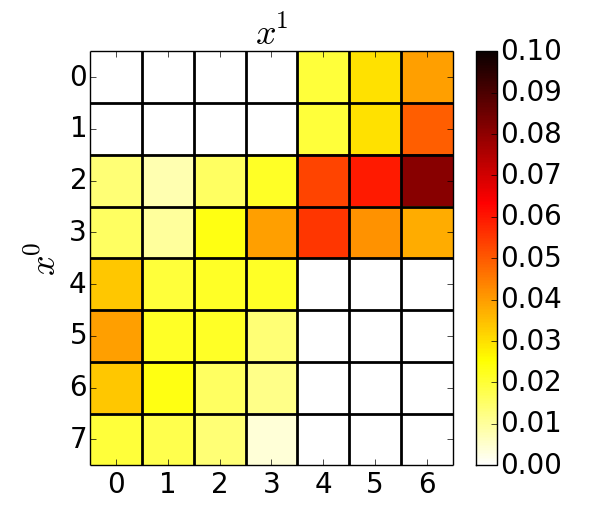}\caption{}
        \label{fig:heatmap}
    \end{subfigure}
\caption{(a) An 8x7 arena. The robot is at $[4,2]$. The dangerous regions are shown in gray. (b) The ratio of time spent in a cell is shown with a heatmap.}\label{fig:toy}
\end{figure}

The dangerous regions are shown in gray in Fig.~\ref{fig:arena}. The robot initially does not know the arena. It senses the dangerous region when it is in it (label 1). We run Alg.~\ref{alg:iterativeapproach} on system~\eqref{eq:toy} with parameters $\overline{oc}=1$, $\underline{oc}=1$, $\overline{p}=1$, ${\underline{val}}=0$, and the datasets include $20$ traces with length $100$.  The algorithm terminates after 6 iterations. The total number of labeled data points, the formula generated by Alg.~\ref{alg:formulasearch} and its $tp$ count are shown in Table~\ref{tab:toy}. 

\begin{table}[h!]
\begin{center}
 \begin{tabular}{|c| c | c | c|c|}
 \hline
 $i $ & {$\mathcal{V}_i$}  & $\Phi_i$  & $tp$ & $fp$ \\ [0.5ex] 
 \hline\hline
  1 & 705 & $ (\pastG_{[1,1]} u = S ) \wedge (\pastF_{[1,1]} ( x^1 > 3  \wedge x^0 > 2) )  $ & 126 & 0  \\ 
 \hline
  2 & 429 & $ (\pastG_{[1,1]} u = N ) \wedge (\pastF_{[1,1]} ( x^1 < 4  \wedge x^0 < 3) )  $ & 125 & 0  \\ 
 \hline
  3 & 109 & $ (\pastG_{[1,1]} u = E ) \wedge (\pastF_{[1,1]} ( x^1 > 2  \wedge x^0 > 3) )  $ & 42 & 0  \\ 
 \hline 
  4 & 77 & $ (\pastG_{[1,1]} u = W ) \wedge (\pastF_{[1,1]} ( x^1 < 5  \wedge x^0 < 2 ) )  $ & 52 & 0  \\ 
 \hline 
  5 & 0 & --   & -- & --  \\ 
 \hline 
  \end{tabular}
\end{center}
\caption{Region avoidance: Violation counts and the optimized cause formulas for each iteration of Alg.~\ref{alg:iterativeapproach}.}\label{tab:toy}
\end{table}

For a trace of the robot controlled by Alg.~\ref{alg:controller} with formula $\Psi=\Phi_1 \vee \Phi_2 \vee \Phi_3 \vee \Phi_4$, the number of visits of a cell over the length of the trace is shown in Fig.~\ref{fig:heatmap} with a heatmap. Due to the randomness in the controller (Alg.~\ref{alg:controller}), the robot can explore the arena without going through the dangerous regions. This simple example shows that the proposed approach can be used to explore an arena and to learn a controller that avoids the bad areas. 

\section{Conclusions and Future Works}\label{sec:conc}
 
 We presented a framework for synthesis of feedback controllers to avoid unwanted events. The main steps are the identification of the controllable causes as a ptSTL formula, and the synthesis of feedback controllers to avoid the satisfaction of the cause formula. The proposed method can be applied to any discrete-time dynamical system with a finite control set. We showed on two examples that the system reduces the number of the unwanted events.  The reduction rate depends on the causality relation between the control inputs and the unwanted events, and the parameters of  the formula synthesis algorithm. A stronger causality relation together with the suitable parameters yields a higher reduction rate. Future research will address verification of the resulting controlled system to  guarantee correctness, and identifying other controllable cause structures.


\bibliography{irmakbib}

\begin{thebibliography}{10}
\providecommand{\url}[1]{#1}
\csname url@samestyle\endcsname
\providecommand{\newblock}{\relax}
\providecommand{\bibinfo}[2]{#2}
\providecommand{\BIBentrySTDinterwordspacing}{\spaceskip=0pt\relax}
\providecommand{\BIBentryALTinterwordstretchfactor}{4}
\providecommand{\BIBentryALTinterwordspacing}{\spaceskip=\fontdimen2\font plus
\BIBentryALTinterwordstretchfactor\fontdimen3\font minus
  \fontdimen4\font\relax}
\providecommand{\BIBforeignlanguage}[2]{{%
\expandafter\ifx\csname l@#1\endcsname\relax
\typeout{** WARNING: IEEEtran.bst: No hyphenation pattern has been}%
\typeout{** loaded for the language `#1'. Using the pattern for}%
\typeout{** the default language instead.}%
\else
\language=\csname l@#1\endcsname
\fi
#2}}
\providecommand{\BIBdecl}{\relax}
\BIBdecl

\bibitem{belta2017}
C.~Belta, B.~Yordanov, and E.~{Aydin Gol}, \emph{Formal Methods for
  Discrete-Time Dynamical Systems}, ser. Studies in Systems, Decision and
  Control.\hskip 1em plus 0.5em minus 0.4em\relax Springer, 2017.

\bibitem{tabuada2006linear}
P.~Tabuada and G.~J. Pappas, ``Linear time logic control of discrete-time
  linear systems,'' \emph{Automatic Control, IEEE Transactions on}, vol.~51,
  no.~12, pp. 1862--1877, 2006.

\bibitem{Kavraki:MPlanning}
A.~Bhatia, L.~E. Kavraki, and M.~Y. Vardi, ``Motion planning with hybrid
  dynamics and temporal goals,'' in \emph{{IEEE} Conference on Decision and
  Control}, Atlanta, GA, 2010, pp. 1108--1115.

\bibitem{zamani2014symbolic}
M.~Zamani, P.~M. Esfahani, R.~Majumdar, A.~Abate, and J.~Lygeros, ``Symbolic
  control of stochastic systems via approximately bisimilar finite
  abstractions,'' \emph{IEEE Transactions on Automatic Control}, vol.~59,
  no.~12, pp. 3135--3150, 2014.

\bibitem{KlBe-HSCC08-book}
M.~Kloetzer and C.~Belta, ``Dealing with non-determinism in symbolic control,''
  in \emph{Hybrid Systems: Computation and Control}, ser. HSCC '08,
  M.~Egerstedt and B.~Mishra, Eds.\hskip 1em plus 0.5em minus 0.4em\relax
  Berlin, Heidelberg: Springer-Verlag, 2008, pp. 287--300.

\bibitem{Simulink}
``Simulink version 8.0, r2012b,'' the MathWorks, Natick, MA, USA.

\bibitem{Bombara:2016}
G.~Bombara, C.-I. Vasile, F.~Penedo, H.~Yasuoka, and C.~Belta, ``A decision
  tree approach to data classification using signal temporal logic,'' in
  \emph{Proceedings of the 19th International Conference on Hybrid Systems:
  Computation and Control}, ser. HSCC '16.\hskip 1em plus 0.5em minus
  0.4em\relax New York, NY, USA: ACM, 2016, pp. 1--10.

\bibitem{Kong:Inference:2014}
Z.~Kong, A.~Jones, A.~Medina~Ayala, E.~Aydin~Gol, and C.~Belta, ``Temporal
  logic inference for classification and prediction from data,'' ser. HSCC
  '14.\hskip 1em plus 0.5em minus 0.4em\relax New York, NY, USA: ACM, 2014, pp.
  273--282.

\bibitem{miningjournal}
X.~Jin, A.~Donzé, J.~V. Deshmukh, and S.~A. Seshia, ``Mining requirements from
  closed-loop control models,'' \emph{IEEE Transactions on Computer-Aided
  Design of Integrated Circuits and Systems}, vol.~34, no.~11, pp. 1704--1717,
  Nov 2015.

\bibitem{Jha2019}
S.~Jha, A.~Tiwari, S.~A. Seshia, T.~Sahai, and N.~Shankar, ``Telex: learning
  signal temporal logic from positive examples using tightness metric,''
  \emph{Formal Methods in System Design}, Jan 2019.

\bibitem{codit2018}
E.~{Aydin Gol}, ``Efficient online monitoring and formula synthesis with past
  stl,'' in \emph{5th IEEE International Conference on Control, Decision and
  Information Technologies (Codit)}, 2018.

\bibitem{acc2019}
A.~Ketenci and E.~Aydin~Gol, ``Synthesis of monitoring rules via data mining,''
  in \emph{2019 American Control Conference (ACC) (accepted, to be
  presented)}.\hskip 1em plus 0.5em minus 0.4em\relax IEEE, 2019.

\bibitem{Bartocci2014}
E.~Bartocci, L.~Bortolussi, and G.~Sanguinetti, ``Data-driven statistical
  learning of temporal logic properties,'' in \emph{FORMATS 2014, LNCS, vol
  8711}.\hskip 1em plus 0.5em minus 0.4em\relax Springer International
  Publishing, 2014, pp. 23--37.

\bibitem{Raman:2015:RSS}
V.~Raman, A.~Donze, D.~Sadigh, R.~M. Murray, and S.~A. Seshia, ``Reactive
  synthesis from signal temporal logic specifications,'' in \emph{Proceedings
  of the 18th International Conference on Hybrid Systems: Computation and
  Control}, ser. HSCC '15.\hskip 1em plus 0.5em minus 0.4em\relax New York, NY,
  USA: ACM, 2015, pp. 239--248.

\bibitem{7447084}
S.~{Sadraddini} and C.~{Belta}, ``Robust temporal logic model predictive
  control,'' in \emph{2015 53rd Annual Allerton Conference on Communication,
  Control, and Computing (Allerton)}, Sep. 2015, pp. 772--779.

\bibitem{coogan2016traffic}
S.~Coogan, E.~A. Gol, M.~Arcak, and C.~Belta, ``Traffic network control from
  temporal logic specifications,'' \emph{IEEE Transactions on Control of
  Network Systems}, vol.~3, no.~2, pp. 162--172, 2016.

\bibitem{Asarin2012}
E.~Asarin, A.~Donz{\'e}, O.~Maler, and D.~Nickovic, ``Parametric identification
  of temporal properties,'' in \emph{RV, 2012. LNCS, vol 7186}.\hskip 1em plus
  0.5em minus 0.4em\relax Springer Berlin Heidelberg, 2012, pp. 147--160.

\end{thebibliography}

\end{document}